\documentclass[reqno]{amsart}
\usepackage[utf8]{inputenc}
\usepackage{pmboxdraw}
\usepackage{amsmath}
\usepackage{amssymb}
\usepackage{xcolor}
\usepackage{graphicx}
\usepackage{float}
\usepackage[hidelinks]{hyperref}
\usepackage[shortlabels]{enumitem}

\calclayout

\usepackage{cite}
\newtheorem{theorem}{Theorem}
\newtheorem*{theoremext}{Theorem}

\newtheorem{definition}{Definition}
\newtheorem{proposition}[theorem]{Proposition}

\newtheorem{corollary}[theorem]{Corollary}

\newcommand{\C}{\mathbb{C}}
\newcommand{\N}{\mathbb{N}}

\newcommand{\R}{\mathbb{R}}

\newcommand{\ler}[1]{\left( #1 \right)}
\newcommand{\lesq}[1]{\left[ #1 \right]}
\newcommand{\lers}[1]{\left\{ #1 \right\}}
\newcommand{\hohc}{\cH \otimes \cH^*}

\newcommand{\abs}[1]{\left| #1 \right|}

\newcommand{\inner}[2]{\left< #1,#2 \right>}

\newcommand{\bra}[1]{\langle #1 |}
\newcommand{\ket}[1]{| #1 \rangle}
\newcommand{\Bra}[1]{\langle\langle #1 ||}
\newcommand{\Ket}[1]{|| #1 \rangle\rangle}

\newcommand{\be}{\begin{equation}}
\newcommand{\ee}{\end{equation}}
\newcommand{\ba}{\begin{array}}
\newcommand{\ea}{\end{array}}

\newcommand{\cH}{\mathcal{H}}
\newcommand{\cB}{\mathcal{B}}
\newcommand{\cK}{\mathcal{K}}
\newcommand{\cP}{\mathcal{P}}

\newcommand{\cS}{\mathcal{S}}
\newcommand{\cC}{\mathcal{C}}

\newcommand{\cT}{\mathcal{T}}

\newcommand{\cL}{\mathcal{L}}

\newcommand{\tr}{\mathrm{tr}}
\newcommand{\dd}{\mathrm{d}}

\newcommand{\cA}{\mathcal{A}}

\newcommand{\sh}{\cS\ler{\cH}}

\author[Gergely Bunth]{Gergely Bunth}
\address{Gergely Bunth, HUN-REN Alfr\'ed R\'enyi Institute of Mathematics\\ Re\'altanoda u. 13-15.\\Budapest H-1053\\ Hungary\\ and Department of Analysis and Operations Research, Institute of Mathematics, Budapest University of Technology and Economics\\ M\H{u}egyetem rkp. 3. \\ Budapest H-1111 \\ Hungary}
\email{bunth.gergely@renyi.hu}

\author[Eszter Szab\'o]{Eszter Szab\'o}
\address{Eszter Szab\'o, ELTE E\"otv\"os Lor\'and University \\ P\'azm\'any P\'eter s\'et\'any 1/A \\ Budapest H-1117 \\ Hungary}
\email{szaeszba@student.elte.hu}

\author[D\'aniel Virosztek]{D\'aniel Virosztek}
\address{D\'aniel Virosztek, HUN-REN Alfr\'ed R\'enyi Institute of Mathematics\\ Re\'altanoda u. 13-15.\\Budapest H-1053\\ Hungary}
\email{virosztek.daniel@renyi.hu}

\date{}



\thanks{Bunth and Virosztek were supported by the Momentum program of the Hungarian Academy of Sciences under grant agreement no. LP2021-15/2021, and by the Hungarian National Research, Development and Innovation Office (NKFIH) under grant agreement no. Excellence\_151232. Virosztek was partially supported also by the ERC Synergy Grant No. 810115.}

\title{Quantum Wasserstein isometries of the $n$-qubit state space: a Wigner-type result}

\begin{document}

\begin{abstract}
We determine the isometry group of the $n$-qubit state space with respect to the quantum Wasserstein distance induced by the so-called symmetric transport cost for all $n \in \N.$ It turns out that the isometries are precisely the Wigner symmetries, that is, the unitary or anti-unitary conjugations. 
\end{abstract}

\maketitle

%================================================================================
%================================================================================
\section{Introduction}

%================================================================================
\subsection{Motivation and main result} \label{subsec:motiv-and-main-res}

The theory of optimal transportation is an area of intense research within the field of analysis, and it is strongly connected to a number of fascinating topics including the geometry of metric measure spaces \cite{LottVillani,Sturm4,Sturm5,Sturm6}, variational analysis \cite{FigalliMaggi1,FigalliMaggi2}, the study of physical evolution equations \cite{JKO-97a,JKO-97b,JKO-98} and stochastic partial differential equations \cite{Hairer2,Hairer3}, and probability theory \cite{bgl,Butkovsky}. The theory was born in 1781 by the seminal work of Monge \cite{monge1781memoire}, and the 1940s have seen groundbreaking contributions to the field by Kantorovich \cite{kantorovich1942-transloc, kantorovich1948-monge}. At the end of the 1980s, research in optimal transport theory gained new momentum when Brenier's work on the structure of optimal transport maps \cite{Brenier-polar-fra,Brenier-polar-en} inspired numerous experts working on fluid mechanics or statistical and mathematical physics to apply optimal transport techniques. 
\par
Quantum mechanical counterparts of the classical optimal transport theory and analogues of the Wasserstein distances have also been proposed in recent decades. We refer to the book \cite{OTQS-book} and the survey papers \cite{beatty2025wasserstein, Trevisan-review-QOT} for a detailed overview of this rapidly evolving field, and we mention here only a few of the quantum optimal transport concepts introduced so far to illustrate the diversity of the topic. A dynamical theory relying on the classical Benamou-Brenier formula \cite{Ben-Bren00} and Jordan-Kinderlehrer-Otto theory \cite{JKO-98} was proposed by Carlen and Maas \cite{CarlenMaas-4, CarlenMaas-3, CarlenMaas-2}, see also the works of Datta, Rouz\'e \cite{DattaRouze1,DattaRouze2} and Wirth \cite{wirth-NC-transport-metric, Wirth-dual}. Caglioti, Golse, Mouhot, and Paul used quantum couplings to define quantum optimal transport problems and induced Wasserstein distances \cite{CagliotiGolsePaul-towardsqot, CagliotiGolsePaul, GolseMouhotPaul, GolsePaul-OTapproach, GolsePaul-wavepackets, GolseTPaul-pseudometrics} while in the De Palma-Trevisan approach \cite{DPT-AHP,DPT-lecture-notes} quantum channels realize the transport --- see also \cite{BPTV-metric-24, BPTV-p-Wass, BPTV-25-Kantorovich, BPTV-Petz-swap, Wirth-triangle} for further developments in this direction. The concept of Friedland, Eckstein, Cole and Życzkowski \cite{Friedland-Eckstein-Cole-Zyczkowski-MK, Cole-Eckstein-Friedland-Zyczkowski-QOT, BistronEcksteinZyczkowski} is also based on couplings, but with strikingly different cost operators. Duvenhage used modular couplings to define quantum Wasserstein distances \cite{Duvenhage1, Duvenhage-quad-Wass-vNA, Duvenhage3, Duvenhage-ext-quantum-det-bal}, and separable quantum Wasserstein distances have also been introduced and studied \cite{TothPitrik, toth-pitrik-2025, beatty-franca}.
\par
Both classical and quantum concepts of optimal transportation induce distances between probability measures resp. quantum states. These distances are called Wasserstein distances and they define an interesting geometric structure on classical and quantum state spaces. In striking contrast to the classical case, quantum Wasserstein distances are not genuine metrics, e.g., states may have a positive distance from themselves. Consequently, maps preserving quantum Wasserstein distances are not necessarily continuous, surjective, or even injective, although a bona fide isometry of a compact metric space has the above three properties. Accordingly, when describing the structure of the Wasserstein distance-preserving maps on the $n$-qubit state spaces in our main result, Theorem \ref{thm:n-qubit-symmetric}, we do not a priori assume injectivity, surjectivity or continuity. 
\par
The structure of maps on quantum state spaces preserving quantum information theoretical distances has been actively studied in recent decades \cite{MolnarPitrikVirosztek-BregmanJensen, Molnar-Bures,Virosztek-quantumchialfa,Virosztek-quantumfpreservers,Virosztek-BregmanJensenPreserver}, and the same holds for classical Wasserstein isometries \cite{BertrandKloeckner-hadamard, BertrandKloeckner-nonnegcurv, Kloeckner-euclidean, Kloeckner-ultrametrics,BertrandKloeckner-boundary,Rodriguez,GTV2020wasserstein-realline,GTV2022wass-hilb, Durham-paper, Balogh-Titkos-Virosztek-Heisenberg, Balogh-Titkos-Virosztek-Carnot}, as well. Still, the first studies on quantum Wasserstein isometries were published only very recently \cite{GPTV23, SV-2025-qubit-QW-isom}. These works concerned the simplest quantum state space, the space of qubits, and the authors characterized its isometries with respect to quantum Wasserstein distances and divergences induced by various distinguished transport costs. An important example is the case of the so-called \emph{symmetric transport cost}, where the observable quantities inducing the cost are the \emph{Pauli operators}. In this case, the quantum Wasserstein isometries are precisely the \emph{Wigner symmetries}, that is, the unitary and anti-unitary conjugations. An exciting phenomenon is that if we remove one Pauli operator from the set of observables, then the structure of isometries becomes much richer \cite{GPTV23}, and even non-surjective and non-injective (!) isometries enter the game --- we note that the quantum Wasserstein distance is not a genuine metric, and that is the reason why non-injective isometries can possibly exist.
\par
The goal of this paper is to study the quantum Wasserstein isometries of the $n$-qubit state spaces for all $n \in \N$ and to give a complete description of the isometry groups when the transport cost is induced by all tensor products of Pauli operators. Our main result is formulated in Theorem \ref{thm:n-qubit-symmetric} below --- all the necessary notions and notations will be introduced in Subsection \ref{subsec:basic-notions-notation}.
\begin{theorem} \label{thm:n-qubit-symmetric}
    A map $\Phi: \cS\ler{\C^{2^n}} \to \cS\ler{\C^{2^n}}$ is a quantum Wasserstein isometry with respect to the symmetric transport cost, that is,  
    \begin{align}
        D_{\mathrm{sym}, n}\ler{\Phi(\rho), \Phi(\omega)}=D_{\mathrm{sym}, n}\ler{\rho,\omega} \qquad \ler{\rho,\omega \in \cS\ler{\C^{2^n}}}
    \end{align}
    if and only if $\Phi$ is a Wigner symmetry, that is, there exists a unitary or anti-unitary operator $U: \C^{2^n} \to \C^{2^n}$ such that $\Phi(\rho)=U\rho U^*$ for all $\rho \in \cS\ler{\C^{2^n}}.$
\end{theorem}

%================================================================================
\subsection{Basic notions, notation} \label{subsec:basic-notions-notation}
As the results presented in this paper are finite-dimensional, we assume for simplicity throughout this work that $\cH$ is a complex finite-dimensional Hilbert space.
Let $\cL(\cH)^{sa}$ denote the set of self-adjoint operators on $\cH,$ and let $\cB(\cH)$ stand for the algebra of all linear operators on $\cH.$ The state space $\cS(\cH)$ on $\cH$ is the set of all positive operators of unit trace on $\cH,$ that is, $\cS(\cH)=\lers{ X \in \cB(\cH) \, \middle| \, X \geq 0, \, \tr_{\cH}\lesq{X}=1}.$ Here and in the following, the relation $\geq$ between self-adjoint operators is the semidefinite (or L\"owner) order. The set of pure states, that is, rank-$1$ projections, on $\cH$ is denoted by $\cP_1(\cH).$ The transpose $A^T$ of a linear operator $A \in \mathrm{Lin}(\cH, \cK)$ is the linear operator mapping $\cK^*$ into $\cH^*$ which is defined by the property that $\varphi(Ax)=\ler{A^T \varphi}(x)$ for all $x \in \cH$ and $\varphi \in \cK^*.$ Here $\cK$ is another finite-dimensional complex Hilbert space, and $\cH^*$ and $\cK^*$ denote the dual spaces of $\cH$ resp. $\cK.$ In this finite-dimensional setting, the support of $A \in \cB(\cH)$ coincides with its range, that is, $\mathrm{supp}(A)=\mathrm{ran}(A).$ The spectrum of an operator $A \in \cB(\cH)$ is denoted by $\mathrm{spec}(A).$ A linear map $\Psi: \cB(\cK) \to \cB(\cH)$ is called a quantum channel if it is completely positive and trace-preserving.
\par 
{\newcommand*{\braket}[2]{\left<#1\right.\left|\,#2\right>}
\newcommand{\Lin}{\text{Lin}}
Here and throughout this work, we will use the canonical linear isomorphism between $\cB(\cH)$ and $\Lin  (\mathbb{C},\cH)\otimes\Lin (\mathbb{C}^*,\cH^*)$ which is the continuous linear extension of the map
\begin{align}
    \ket{\psi}\bra{\varphi} \mapsto \ket{\psi} \otimes \bra{\varphi}^T \qquad \ler{\psi, \varphi \in \cH}. \nonumber
\end{align}
That is, for $\sum_{i}\lambda_i\ket{a_i}\bra{a_i}=A\in \cB(\cH)$ and any two orthonormal bases $\mathcal{E},\mathcal{F}$ in $\cH,$ the image $\Ket{A}$ of $A$ under this canonical isomorphism is given in the following basis-independent way:
$$
\ket{\ket{A}}:=\sum_{i,j}\braket{e_i}{Af_j}\ket{e_i}\otimes\bra{f_j}^T=\sum_{j}\ket{Af_j}\otimes\bra{f_j}^T=\sum_{i}\ket{e_i}\otimes\bra{A^*e_i}^T
$$
$$
\hspace{-10pt}=\sum_{i,j} \braket{e_i}{Ae_j}\ket{e_i}\otimes \bra{e_j}^T=\sum_i \lambda_i\ket{a_i}\otimes\bra{a_i}^T.
$$
Up to trivial isomorphisms, $\Ket{A}$ can be viewed as an element of $\Lin(\mathbb{C},\hohc)$. Then $\Bra{A}$ is defined naturally as $\Bra{A}=\Ket{A}^*$ with respect to the Hilbert-Schmidt inner product, that is, 
$$
\Bra{A}:=\sum_i \bar{\lambda_i}\bra{a_i}\otimes\ket{a_i}^T.
$$
It is immediate from this definition that

\begin{align} \label{eq:product-rule}
    B\otimes C^T\ket{\ket{A}}=\ket{\ket{BAC}} \, \text{ for all } A, B, C \in \cB(\cH), 
\end{align}  
and that 
\begin{align} \label{eq:canonical-purification}
\ket{\ket{\sqrt{\rho}}}\bra{\bra{\sqrt{\rho}}}=\sum_{i,j}\sqrt{\lambda_i\lambda_j}\ket{r_i}\bra{r_j}\otimes\left(\ket{r_i}\bra{r_j}\right)^T
\end{align}
is a basis-independent purification, the so-called canonical purification \cite{Holevo} of $\rho=\sum_{i}\lambda_i\ket{r_i}\bra{r_i}$ --- see also Lemma 1 and Definition 2 in \cite{DPT-AHP}.
}
\par
The classical quadratic optimal transport problem \cite{Ambrosio-otbook1, Ambrosio-otbook2, Figalli-book, Villani1, Villani2} on a complete and separable metric space $(X,d)$ is the following: given Borel probability measures $\mu$ and $\nu$ on $(X,d)$ with finite second moment, the goal is to minimize the total transport cost
\begin{align}
    \mathrm{Cost}(\pi)=\iint_{X \times X} d^2(x,y) \dd \pi(x,y) \nonumber
\end{align}
among all admissible transport plans $\pi.$ The admissible transport plans between $\mu$ and $\nu$ are exactly the couplings of $\mu$ and $\nu,$ that is, those probability measures on $X \times X$ that admit $\mu$ and $\nu$ as first resp. second marginal. The squared quadratic Wasserstein distance between $\mu$ and $\nu$ is defined as the optimal quadratic transport cost, that is,
\begin{align}
    d_{W_2}^2(\mu,\nu)=\inf\lers{\iint_{X \times X} d^2(x,y) \dd \pi(x,y) \, \middle| \, \pi \in \Gamma(\mu, \nu)},
\end{align}
where $\Gamma(\mu, \nu)$ stands for the set of all couplings of $\mu$ and $\nu.$
\par
In \cite{DPT-AHP}, De Palma and Trevisan introduced a concept of optimal transportation in quantum mechanics where quantum channels realize the transport the following way: 
the transport plans between the states $\rho$ and $\omega$ in $\cS(\cH)$ are quantum channels $\Phi: \cB\ler{\mathrm{supp}(\rho)} \to \cB(\cH)$ sending $\rho$ to $\omega,$ and a transport plan $\Phi$ gives rise to the quantum coupling $\Pi_{\Phi}$ the following way: 
\begin{align}
\label{eq:Pi-phi-def}
\Pi_{\Phi}=\ler{\Phi \otimes \mathrm{id}_{\mathcal{T}_1\ler{\cH^*}}} \ler{\Ket{\sqrt{\rho}}\Bra{\sqrt{\rho}}},
\end{align} 
where $\Ket{\sqrt{\rho}}\Bra{\sqrt{\rho}} \in \cS\ler{\hohc}$ is the canonical purification (see \eqref{eq:canonical-purification} above and also \cite{Holevo}) of the state $\rho \in \sh.$ 
It is easy to check that $\Pi_{\Phi}$ is a state on $\hohc$ such that its first marginal is $\omega$ while the second marginal is $\rho^T,$ that is, $\tr_{\cH^*}\lesq{\Pi_{\Phi}}=\omega$ and $\tr_{\cH}\lesq{\Pi_{\Phi}}=\rho^T.$ Therefore, the set of all quantum couplings of the states $\rho,\omega \in \sh$ (denoted by $\cC(\rho, \omega)$) was defined in \cite{DPT-AHP} by
\be \label{eq:q-coup-def}
\cC\ler{\rho, \omega}=\lers{\Pi \in \cS\ler{\cH \otimes \cH^*} \, \middle| \, \tr_{\cH^*} [\Pi]=\omega, \,  \tr_{\cH} [\Pi]=\rho^T}.
\ee
In other words, a coupling of $\rho$ and $\omega$ is a state $\Pi$ on $\hohc$ such that 
\be \label{eq:part-trace-def}
\tr_{\hohc}[\ler{A\otimes I_{\cH}^T} \Pi]=\tr_{\cH} [\omega A] \text{ and }
\tr_{\hohc}\lesq{\ler{I_{\cH} \otimes B^{T}} \Pi}=\tr_{\cH^*} [\rho^T B^T]=\tr_{\cH} [\rho B]
\ee
for all $A, B \in \cB(\cH).$ Note that if either $\rho$ or $\omega$ is a pure state, that is, a rank-one projection, then $\cC(\rho,\omega)$ is a singleton containing only the product (or independent) coupling $\omega \otimes \rho^T.$
\par
Given a finite collection of observable quantities $\cA=\lers{A_1, \dots, A_K},$ where $A_k \in \cL(\cH)^{sa}$ for all $k,$ the quadratic transportation cost of a coupling $\Pi \in \cC(\rho, \omega)$ in this finite-dimensional setting is defined the following way:
\begin{align} \label{eq:cost-in-fin-dim-and-quad-cost-op}
    \mathrm{Cost}_{\cA}(\Pi)=\tr_{\hohc}\lesq{\Pi C_{\cA}}, \text{ where } C_{\cA}=\sum_{k=1}^K \ler{A_k \otimes I^T - I \otimes A_k^T}^2. 
\end{align}
The operator $C_{\cA}$ defined by $\cA=\{A_1, \dots, A_K\}$ as in \eqref{eq:cost-in-fin-dim-and-quad-cost-op} is called the (quadratic) transport cost operator corresponding to $\cA.$
Moreover, the cost of the coupling $\Pi_\Phi$ corresponding to the channel $\Phi$ (see \eqref{eq:Pi-phi-def}) can be expressed referring to the channel instead of the coupling \cite[Sec. 5.2]{DPT-lecture-notes}:
\begin{align} \label{eq:cost-in-terms-of-channel}
    \mathrm{Cost}_{\cA}\ler{\Pi_\Phi}=\sum_{k=1}^K \ler{\tr_{\cH}\lesq{\Phi(\rho) A_k^2}
    +\tr_{\cH}\lesq{\rho A_k^2} -2 \tr_{\cH}\lesq{\sqrt{\rho}A_k\sqrt{\rho}\Phi^{\dagger}(A_k)}},
\end{align}
where $\Phi^{\dagger}: \cB(\cH) \to \cB\ler{\mathrm{supp}(\rho)}$ is the adjoint of the channel $\Phi: \cB\ler{\mathrm{supp}(\rho)} \to \cB(\cH)$ defined by the property that $\tr_{\cH}\lesq{\Phi(X)A}=\tr_{\mathrm{supp}(\rho)}\lesq{X \Phi^{\dagger}(A)}$ for all $X \in \cB\ler{\mathrm{supp}(\rho)}$ and $A \in \cB(\cH).$ 
In complete analogy with the classical case, optimal solutions of quantum transport problems give rise to quantum Wasserstein distances.

\begin{definition} \label{def:QW-distances-quadratic}
The quadratic quantum Wasserstein distance of the states $\rho, \omega \in \cS(\cH)$ with respect to the finite collection of observables $\cA=\lers{A_1, \dots, A_K}$ is denoted by $D_{\cA}$ and is defined as the square root of the minimal quadratic transport cost: 
\begin{align} \label{eq:quadratic-QW-dist-def}
    D_{\cA}^2(\rho, \omega)= \inf \lers{ \mathrm{Cost}_{\cA}(\Pi) \, \middle| \, \Pi \in \cC(\rho, \omega)}.
\end{align}
\end{definition}

Our work concerns $n$-qubit state spaces, and hence the tensor products of Pauli operators will play a central role.
The Pauli operators act on $\C^2$ and are defined by
\begin{align} \label{eq:Pauli-def}
\sigma_0=I=\lesq{\ba{cc}1&0\\0&1 \ea}, \,
\sigma_1=\sigma_x=\lesq{\ba{cc}0&1\\1&0 \ea}, \,
\sigma_2=\sigma_y=\lesq{\ba{cc}0&-i\\i&0 \ea}, \,
\sigma_3=\sigma_z=\lesq{\ba{cc}1&0\\0&-1 \ea}.
\end{align}
In \cite{GPTV23}, the authors considered the so-called \emph{symmetric} transport cost governed by the observables $\cA=\lers{\sigma_0, \sigma_1, \sigma_2, \sigma_3}$ which give rise to the $1$-qubit symmetric transport cost operator 
\begin{align} \label{eq:1-qubit-symm-cost-op-def}
C_{\text{sym},1}:=C_{\lers{\sigma_0, \sigma_1, \sigma_2, \sigma_3}}=\sum_{j=0}^3\left(\sigma_j\otimes I^T-I\otimes \sigma_j^T\right)^2.
\end{align}
Note that $\sigma_0=I$ does not contribute to the cost operator $C_{\text{sym},1}$ and hence it could have been omitted. Nevertheless, for symmetry reasons, it is more beneficial to keep it in $\cA.$
Our current work focuses on the transport cost which is generated by the collection of all possible tensor products of Pauli operators. That is, we consider the case 
\begin{align}
    \cA=\lers{ \bigotimes_{i=1}^n \sigma_{m(i)} \, \middle| \, m \in \{0,1,2,3\}^{\{1,2, \dots n\}} }. \nonumber
\end{align}
We denote the induced transport cost operator by $C_{\text{sym},n},$ which is given by
\begin{align} \label{eq:n-qubit-symmetric-cost-op-def}
    C_{\text{sym},n}=\sum_{m \in \{0,1,2,3\}^{\{1,2, \dots n\}}} \ler{\bigotimes_{i=1}^n \sigma_{m(i)} \otimes I_n^T - I_n \otimes \ler{\bigotimes_{i=1}^n \sigma_{m(i)}}^T}^2,
\end{align}
where $I_n$ stands for the identity map on $\C^{2^n}.$
\par
Throughout this work, by isometry we mean simply a distance-preserving transformation, and we do not assume surjectivity, injectivity or continuity. Clearly, if the distance is a bona fide metric (and it is important to note that quantum Wasserstein distances are not!) then every isometry is injective and continuous. 
As mentioned in Subsection \ref{subsec:motiv-and-main-res}, even on the qubit state space there exist exotic (in the sense that they are non-injective, non-continuous, and non-surjective) maps that preserve certain quantum Wasserstein distances --- see Theorem 2 in \cite{SV-2025-qubit-QW-isom} and Theorem 2 in \cite{GPTV23}.

 %================================================================================
%================================================================================
\section{The spectral resolution of the $n$-qubit symmetric transport cost operator}

A vital part of understanding the structure of quantum Wasserstein isometries is the clear understanding of the spectral properties of the symmetric transport cost operator $C_{\text{sym},n}$ given by \eqref{eq:n-qubit-symmetric-cost-op-def}. This section is devoted to this key ingredient of the proof of Theorem \ref{thm:n-qubit-symmetric}. The following Proposition concerns quadratic transport cost operators generated by two-level observables, and asserts in particular that $C_{\text{sym}, n}$ is a constant multiple of a projection of rank $2^{2n}-1.$

%================================================================================
%\subsection{The qubit case}
{%\color{blue}
\newcommand{\set}[1]{\left\{#1\right\}}
\newcommand{\Lin}{\text{Lin}}
\begin{proposition} \label{prop:spectral-resolution}
    If all observables in $\cA=\set{A_k}_{k=1}^K \subset \cL(\cH)^{sa}$ are two-level observables with symmetric spectrum, that is, $A_k=\lambda_kP_k-\lambda_k(I-P_k)$, where $P_k$ is an ortho-projection on $\cH$ and $\lambda_k \in \R,$ then
    \be\label{eq:cost_2level}
        C_{\cA}=\sum_{k=1}^K\ler{A_k\otimes I_{\cH}^T-I_{\cH}\otimes A_k^T}^2
        =\sum_{k=1}^K 2\ler{\lambda_k^2 I_{\cH}\otimes I_{\cH}^T-A_k\otimes A_k^T}.
    \ee
    Equivalently, if all observables in $\cA=\set{\bigotimes_{i=1}^nA_{k,i}}_{k=1}^K \subset \cL\ler{\cH^{\otimes n}}^{sa}$ are tensor products of two-level observables on $\cH$ such that $\mathrm{spec}\ler{A_{k,i}}=\lers{-\lambda_k, \lambda_k}$ for all $k$ and $i,$ then
    \be\label{eq:cost_2level_tensor}
        C_{\cA}=\sum_{k=1}^K 2\left(\lambda_k^{2n} I_{\cH^{\otimes n}}\otimes I_{\cH^{\otimes n}}^T-\bigotimes_{i=1}^n A_{k,i}\otimes \bigotimes_{i=1}^n A_{k,i}^T \right).
    \ee
    In particular, if $\cA=\set{\bigotimes_{i=1}^n\sigma_{k,i}}_{k=1}^K \subset \cL\ler{\C^{2^n}}^{sa}$ consists only of tensor products of Pauli operators, then
    \be\label{eq:cost_Pauli_tensor}
        C_{\cA}=\sum_{k=1}^K 2\left(I_n\otimes I_n^T-\bigotimes_{i=1}^n\sigma_{k,i}\otimes \bigotimes_{i=1}^n\sigma_{k,i}^T \right),
    \ee
    and, as a consequence, all such cost operators commute with each other and have a common eigenbasis, the basis formed by all tensor products of Pauli operators viewed as elements of $\C^{2^n} \otimes \ler{\C^{2^n}}^*.$
    If furthermore $\cA=\set{\bigotimes_{i=1}^n\sigma_{k,i}}_{k=1}^K \subset \cL\ler{\C^{2^n}}^{sa}$ contains each of the $4^n$ distinct tensor products of Pauli operators exactly once, then
    \be\label{eq:n-qubit-symmetric-cost-spectral}
        C_{\cA}=C_{\text{sym},n}=2^{2n+1}I_n\otimes I_n^T-2^{n+1} \ket{\ket{I_n}}\bra{\bra{I_n}}=2^{2n+1}\left(I_n\otimes I_n^T-\frac{1}{2^n}\ket{\ket{I_n}}\bra{\bra{I_n}}\right).
    \ee
\end{proposition}
\begin{proof}
    We prove \eqref{eq:cost_2level} first. Clearly,
    \begin{align}
        C_{\cA}&=\sum_{k=1}^K\ler{\left(\lambda_kP_k-\lambda_k(I_{\cH}-P_k)\right)\otimes I_{\cH}^T-I_{\cH}\otimes \left(\lambda_k P_k-\lambda_k(I_{\cH}-P_k)\right)^T}^2 \nonumber \\
         &=\sum_{k=1}^K\lambda_k^2\ler{\left(P_k-(I_{\cH}-P_k)\right)\otimes I_{\cH}^T-I_{\cH}\otimes \left(P_k-(I_{\cH}-P_k)\right)^T}^2 \nonumber \\
         &=\sum_{k=1}^K4\lambda_k^2 \left[\frac{A_k+\lambda_k I_{\cH}}{2\lambda_k}\otimes\left(I_{\cH}-\frac{A_k+\lambda_k I_{\cH}}{2\lambda_k}\right)^T+\ler{I_{\cH}-\frac{A_k+\lambda_k I_{\cH}}{2\lambda_k}}\otimes \left(\frac{A_k+\lambda_k I_{\cH}}{2\lambda_k}\right)^T\right] \nonumber \\
         &=\sum_{k=1}^K 2 \ler{\lambda_k^2 I_{\cH}\otimes I_{\cH}^T-A_k\otimes A_k^T}. \nonumber
    \end{align}
    The equivalence of \eqref{eq:cost_2level} and \eqref{eq:cost_2level_tensor} is straightforward as $\bigotimes_{i=1}^n A_{k,i}$ is a two-level observable with symmetric spectrum if every $A_{k,i}$ is so, and one gets \eqref{eq:cost_2level} from \eqref{eq:cost_2level_tensor} as the special case $n=1.$  Eq. \eqref{eq:cost_Pauli_tensor} is a special case of \eqref{eq:cost_2level_tensor}. The cost operators of the form \eqref{eq:cost_Pauli_tensor} contain only tensor squares of Pauli product operators, and hence these cost operators commute with each other. Alternatively, one can see directly that the basis formed by all tensor products of Pauli operators is a common eigenbasis of costs of the form \eqref{eq:cost_Pauli_tensor}, that is, for any tensor product of Pauli operators $\bigotimes_{l=i}^n\sigma_{m'(i)}$ one gets
    \begin{align}
        &\sum_{k=1}^K2\left(I\otimes I^T-\bigotimes_{i=1}^n\sigma_{m(k,i)}\otimes \bigotimes_{i=1}^n\sigma_{m(k,i)}^T\right)\ket{\ket{\bigotimes_{i=1}^n\sigma_{m'(i)}}} \nonumber \\
        =&\sum_{k=1}^K2\left(1-\prod_{i=1}^nj_{m(k,i),m'(i)}\right)\ket{\ket{\bigotimes_{i=1}^n\sigma_{m'(i)}}}, \nonumber
    \end{align} 
    where
    \begin{align}
        j_{m,m'}:=
        \begin{cases}
            1 & \text{if } m=0, \text{ or } m'=0, \text{ or }m=m',\\
            -1 & \text{if } m\neq 0 \text{ and } m'\neq 0, \text{ and }m\neq m'.
        \end{cases} \nonumber
    \end{align}
    In the above computation we used that 
    \begin{align}
    \left(\sigma_{m}\otimes\sigma_{m}^T\right)\ket{\ket{\sigma_{m'}}}=j_{m,m'}\ket{\ket{\sigma_{m'}}}, \nonumber
    \end{align}
    which follows from the identity \eqref{eq:product-rule} and from the fact that Pauli operators either commute or anti-commute, and two Pauli operators commute if and only if they coincide or one of them is $\sigma_0=I_{\C^2}.$ Now if $\cA=\set{\bigotimes_{i=1}^n\sigma_{m(k,i)}}_{k=1}^K$ contains each of the $4^n$ ($I_{\C^2}^{\otimes n}$ included) distinct Pauli product terms exactly once, then 
    \begin{align}
        \sum_{k=1}^K2\left(1-\prod_{i=1}^nj_{m(k,i),{m'(i)}}\right)\ket{\ket{\bigotimes_{i=1}^n\sigma_{m'(i)}}}
        &=\left(2^{2n+1}-2\prod_{i=1}^n\left(\sum_{m=0}^3 j_{m,{m'(i)}}\right)\right)\ket{\ket{\bigotimes_{i=1}^n\sigma_{m'(i)}}}
         \nonumber \\      =\left(2^{2n+1}-2\prod_{i=1}^n\left(4\delta_{0,m'(i)}\right)\right)\ket{\ket{\bigotimes_{i=1}^n\sigma_{m'(i)}}}
        &=2^{2n+1}\left(1-\prod_{i=1}^n\delta_{0,m'(i)}\right)\ket{\ket{\bigotimes_{i=1}^n\sigma_{m'(i)}}} \nonumber\\
        &=2^{2n+1}\left(I\otimes I^T-\frac{1}{2^n}\ket{\ket{I}}\bra{\bra{I}}\right)\ket{\ket{\bigotimes_{i=1}^n\sigma_{m'(i)}}}. \nonumber
    \end{align}
    Note finally that the maps
    \begin{align}
        \set{\ket{\ket{\bigotimes_{i=1}^n\sigma_{m'(i)}}}}_{m'\in \set{0,1,2,3}^n} \nonumber
    \end{align}
    form a generator system in $\Lin(\C,\C^{2^n})\otimes\Lin(\C^*,(\C^{2^n})^*)$, which is in fact an orthogonal basis, from which \eqref{eq:n-qubit-symmetric-cost-spectral} follows.
\end{proof}
}

%================================================================================
%================================================================================
\section{Quantum Wasserstein isometries of the $n$-qubit state space with respect to the symmetric transport cost --- the proof of Theorem \ref{thm:n-qubit-symmetric}}

The proof of Theorem \ref{thm:n-qubit-symmetric} is divided into several steps. First, we prove that every Wigner symmetry (that is, unitary or anti-unitary conjugation) of the state space $\cS\ler{\C^{2^n}}$ is a quantum Wasserstein isometry. We recall that a linear map $U: \cH \to \cH$ is called unitary if it is surjective and $\inner{Ux}{Uy}=\inner{x}{y}$ for all $x,y \in \cH,$ while a conjugate-linear map $U: \cH \to \cH$ is called anti-unitary if it is surjective and $\inner{Ux}{Uy}=\inner{y}{x}=\overline{\inner{x}{y}}$ for all $x,y \in \cH.$ We also note that the transpose $A^T$ of a bounded conjugate-linear map $A: \cH \to \cK$ is the bounded conjugate-linear map $A^T$ that maps $\cK^*$ into $\overline{\cH}^*$ and is defined by the property  that $\varphi(Ax)=\ler{A^T \varphi}(x)$ for all $x \in \cH$ and $\varphi \in \cK^*,$ where $\overline{\cH}^*$ stands for the space of conjugate-linear functionals on $\cH.$ The precise statement is formulated in Proposition \ref{prop:Wigner-symmetries-are-isometries} below. The proof of Proposition \ref{prop:Wigner-symmetries-are-isometries} follows a pattern which is similar to the proof of the analogous statement for qubit in \cite{GPTV23}. However, the precise understanding of the spectral decomposition of the cost operator $C_{\text{sym},n}$ allows us to greatly simplify the computations, even though the $1$-qubit case discussed in \cite{GPTV23} is a very special case of the $n$-qubit case treated here.

\begin{proposition} \label{prop:Wigner-symmetries-are-isometries}
Let $\rho, \omega \in \cS\ler{\C^{2^n}}$ and $U$ be a unitary or anti-unitary transformation on $\C^{2^n}.$ Then 
\begin{align} \label{eq:unitary-invariance-of-D-sym-n}
    D_{\text{sym},n}(U \rho U^*, U \omega U^*)= D_{\text{sym},n}(\rho, \omega).
\end{align}
\end{proposition}

\begin{proof}
    It is sufficient to show that 
    \begin{align} \label{eq:unitary-decreasing-D-sym-n}
        D_{\text{sym},n}(U \rho U^*, U \omega U^*) \leq D_{\text{sym},n}(\rho, \omega) \qquad \ler{\rho, \omega \in \cS\ler{\C^{2^n}}}
    \end{align}
    for all unitary and anti-unitary $U,$ because the inverse of a(n anti-)unitary conjugation is again a(n anti-)unitary conjugation, and hence if \eqref{eq:unitary-decreasing-D-sym-n} holds for all $U,$ then the reversed inequality $D_{\text{sym},n}(U \rho U^*, U \omega U^*) \geq D_{\text{sym},n}(\rho, \omega)$ also holds as 
    \begin{align}
        D_{\text{sym},n}(\rho, \omega)
        =D_{\text{sym},n}(U^*(U \rho U^*)U, U^*(U \omega U^*)U) 
        \leq D_{\text{sym},n}(U \rho U^*, U \omega U^*).
        \nonumber
    \end{align}
    To prove \eqref{eq:unitary-decreasing-D-sym-n}, we first justify that
    \begin{align} \label{eq:inclusion-of-couplings}
        \ler{U\otimes (U^*)^T} \cC\ler{\rho, \omega} \ler{U\otimes (U^*)^T}^{*} \subseteq \cC\ler{U\rho U^*, U \omega U^*}
    \end{align}
    holds for all (anti-)unitary $U$ acting on $\C^{2^n}.$ Let $\Pi \in \cC(\rho, \omega)$ and assume that $\Pi=\sum_{l=1}^L A_l \otimes B_l^T$ for some $A_1, \dots, A_L, B_1, \dots, B_L \in \cB\ler{\C^{2^n}}.$ 
    Then 
    \begin{align}
        \ler{U\otimes (U^*)^T} \Pi \ler{U\otimes (U^*)^T}^{*}
        =\sum_{l=1}^L U A_l U^* \otimes (U^*)^T B_l^T U^T, \nonumber
    \end{align}
    and hence 
    \begin{align}
        \tr_{\ler{\C^{2^n}}^*} \lesq{\ler{U\otimes (U^*)^T} \Pi \ler{U\otimes (U^*)^T}^{*}}
        &=\sum_{l=1}^L U A_l U^* \tr_{\ler{\C^{2^n}}^*}\lesq{\ler{U B_l U^*}^T} 
        =\sum_{l=1}^L U A_l U^* \tr_{\ler{\C^{2^n}}^*} \lesq{B_l^T} \nonumber \\
        &=U \ler{\sum_{l=1}^L A_l \tr_{\ler{\C^{2^n}}^*} \lesq{B_l^T}} U^* 
        =U \ler{\tr_{\ler{\C^{2^n}}^*}[\Pi]} U^*
        =U \omega U^*. \nonumber
    \end{align}
    which shows that the first marginal of $\ler{U\otimes (U^*)^T} \Pi \ler{U\otimes (U^*)^T}^{*}$ is indeed $U \omega U^*.$ The second marginal condition
    $\tr_{\C^{2^n}}\lesq{\ler{U\otimes (U^*)^T} \Pi \ler{U\otimes (U^*)^T}^{*}}=\ler{U \rho U^*}^T$ can be justified similarly, and the positive semidefinite property of $\ler{U\otimes (U^*)^T} \Pi \ler{U\otimes (U^*)^T}^{*}$ is clear as (anti-)unitary conjugations preserve the positivity of operators. We used that taking the adjoint and taking the transpose commute, that is, $\ler{X^*}^T=\ler{X^T}^*$ for any $X \in \cB(\cH),$ and we shall use this identity also in the sequel. The next step is to prove that 
    $\ler{U\otimes (U^*)^T}^* C_{\text{sym},n}\ler{U\otimes (U^*)^T}= C_{\text{sym},n}$ for every $U.$ Indeed, by \eqref{eq:n-qubit-symmetric-cost-spectral} one gets 
    \begin{align} \label{eq:unitary-invariance-of-C-sym-n}
        \ler{U\otimes (U^*)^T}^* C_{\text{sym},n}\ler{U\otimes (U^*)^T}
        =\ler{U\otimes (U^*)^T}^* \ler{2^{2n+1}I_n\otimes I_n^T-2^{n+1} \ket{\ket{I_n}}\bra{\bra{I_n}}}\ler{U\otimes (U^*)^T} \nonumber \\
        =2^{2n+1}I_n\otimes I_n^T-2^{n+1} \ket{\ket{U^* I_n U}}\bra{\bra{ U^* I_n U}}
        =2^{2n+1}I_n\otimes I_n^T-2^{n+1} \ket{\ket{I_n}}\bra{\bra{I_n}}=C_{\text{sym},n},
    \end{align}
    where we used the identity \eqref{eq:product-rule}.
    Now we can rely on \eqref{eq:inclusion-of-couplings} and \eqref{eq:unitary-invariance-of-C-sym-n} and justify \eqref{eq:unitary-decreasing-D-sym-n} as follows:
    \begin{align}
        D_{\text{sym},n}^2(U \rho U^*, U \omega U^*)
        &=\inf \lers{\tr_{\C^{2^n} \otimes \ler{\C^{2^n}}^*}\lesq{\Gamma C_{\text{sym},n}} \, \middle| \Gamma \in \cC\ler{U\rho U^*, U \omega U^*}} \nonumber \\
        &\leq \inf \lers{\tr_{\C^{2^n} \otimes \ler{\C^{2^n}}^*}\lesq{\ler{U\otimes (U^*)^T} \Pi \ler{U\otimes (U^*)^T}^* C_{\text{sym},n}} \, \middle| \Pi \in \cC\ler{\rho, \omega}}  \nonumber \\
        &=\inf \lers{\tr_{\C^{2^n} \otimes \ler{\C^{2^n}}^*}\lesq{\Pi \ler{U\otimes (U^*)^T}^* C_{\text{sym},n}\ler{U\otimes (U^*)^T}} \, \middle| \Pi \in \cC\ler{\rho, \omega}} \nonumber\\
        &=\inf \lers{\tr_{\C^{2^n} \otimes \ler{\C^{2^n}}^*}\lesq{\Pi  C_{\text{sym},n}} \, \middle| \Pi \in \cC\ler{\rho, \omega}} = D_{\text{sym},n}^2(\rho, \omega).
    \end{align}
    Therefore, \eqref{eq:unitary-invariance-of-D-sym-n} holds, as desired.
\end{proof}

Now we start proving that only the Wigner symmetries act isometrically on $\cS\ler{\C^{2^n}}$ with respect to the quantum Wasserstein distance $D_{\text{sym}, n}.$ As a first step in this direction, we compute the diameter of $\cS\ler{\C^{2^n}}$ equipped with the Wasserstein distance $D_{\text{sym}, n}$ and characterize those cases when this diameter is realized. 

\begin{proposition} \label{prop:diameter-of-S-n}
For all $n \in \N$ we have 
\begin{align}
    \mathrm{diam}\ler{\cS\ler{\C^{2^n}},D_{\text{sym}, n}}
    =\sup \lers{D_{\text{sym}, n}(\rho,\omega) \, \middle| \, \rho, \omega \in \cS\ler{\C^{2^n}}}
    =2^{n+1/2},     \nonumber
\end{align}
and $D_{\text{sym}, n}(\rho,\omega)=2^{n+1/2}$ if and only if the independent coupling $\omega \otimes \rho^T$ of $\rho$ and $\omega$ is optimal with respect to the cost $C_{\text{sym}, n},$ and $\rho$ and $\omega$ are orthogonal in the Hilbert-Schmidt sense, that is, $\tr_{\C^{2^n}}[\rho \, \omega]=0.$
    
\end{proposition}

\begin{proof}
    According to Proposition \ref{prop:spectral-resolution} and in particular eq. \eqref{eq:n-qubit-symmetric-cost-spectral} there, $C_{\text{sym},n}=2^{2n+1}I_n\otimes I_n^T-2^{n+1} \ket{\ket{I_n}}\bra{\bra{I_n}}$ and hence the cost of the independent coupling of any $\rho, \omega \in \cS\ler{\C^{2^n}}$ takes the simple form 
    \begin{align} \label{eq:indep-coupling-cost}
        \tr_{\C^{2^n} \otimes \ler{\C^{2^n}}^*}\lesq{\omega \otimes \rho^T C_{\text{sym},n}}
        &=\tr_{\C^{2^n} \otimes \ler{\C^{2^n}}^*}\lesq{\omega \otimes \rho^T \ler{2^{2n+1}I_n\otimes I_n^T-2^{n+1} \ket{\ket{I_n}}\bra{\bra{I_n}}}} \nonumber \\
        &=2^{2n+1}-2^{n+1}\tr_{\C^{2^n}}[\rho \, \omega], %\nonumber
    \end{align}
    where we used the identity \eqref{eq:product-rule}.
    Consequently, by the definition of the quadratic quantum Wasserstein distance, 
    \begin{align} \label{eq:dist-upper-bound}
        D_{\text{sym},n}^2 (\rho, \omega) &= \inf \lers{ \tr_{\C^{2^n} \otimes \ler{\C^{2^n}}^*}\lesq{\Pi C_{\text{sym},n}} \, \middle| \, \Pi \in \cC(\rho, \omega)} \nonumber \\
        &\leq \tr_{\C^{2^n} \otimes \ler{\C^{2^n}}^*}\lesq{\omega \otimes \rho^T C_{\text{sym},n}}
        =2^{2n+1}-2^{n+1}\tr_{\C^{2^n}}[\rho \, \omega] \leq 2^{2n+1}
    \end{align}
    holds for any $\rho, \omega \in \cS\ler{\C^{2^n}},$ and the first inequality in \eqref{eq:dist-upper-bound} is saturated if and only if $\omega \otimes \rho^T$ is optimal, while the second inequality is saturated if and only if $\tr_{\C^{2^n}}[\rho \, \omega]=0.$ Here we used that $\tr_{\cH}\lesq{XY} \geq 0$ for every positive (semidefinite) $X,Y \in \cT_2(\cH).$
\end{proof}

The next step is a metric characterization of pure states. 

\begin{proposition} \label{prop:metric-char-of-pure-states}
    For a state $\rho \in \cS\ler{\C^{2^n}}$ the following are equivalent: 
    \begin{itemize}
        \item[(i)] \label{item:rho-is-pure} $\rho \in \cP_1\ler{\C^{2^n}},$ that is, $\rho$ is a pure state;
        \item[(ii)] \label{item:rho-has-partners} there exist states $\rho_1, \rho_2, \dots, \rho_{2^n-1} \in \cS\ler{\C^{2^n}}$ such that $D_{\text{sym},n}(\rho_j,\rho_k)=\mathrm{diam}\ler{\cS\ler{\C^{2^n}},D_{\text{sym}, n}}=2^{n+1/2}$ for all $j,k \in \lers{0,1,\dots, 2^n-1}, \, j \neq k,$ where $\rho_0=\rho.$
    \end{itemize}
\end{proposition}

\begin{proof}

We start with the direction (i) $\Rightarrow$ (ii). If $\varrho\in\mathcal{P}_1\ler{\mathbb{C}^{2^n}}$, then $\varrho=\ket{\psi}\bra{\psi}$ with a unit vector $\psi\in\mathbb{C}^{2^n}.$
We can consider an orthonormal basis $\{\psi_0:=\psi, \psi_1,\ldots,\psi_{2^n-1}\}$ in $\mathbb{C}^{2^n}$ and the corresponding projections of rank $1$:
\begin{align}
\varrho_j:=\ket{\psi_j}\bra{\psi_j},\quad \forall j\in\{0,1,\ldots,2^n-1\}. \nonumber  
\end{align}
Pure states admit only one coupling, the product coupling, that is, $\cC(\rho_j, \rho_k)=\lers{\rho_k\otimes \rho_j^T},$ and hence 
\begin{align} %\label{eq:D-sym-distance-on-pure-states}
    D_{\text{sym},n}^2(\rho_j,\rho_k)
    =\tr_{\C^{2^n} \otimes \ler{\C^{2^n}}^*}\lesq{\rho_k \otimes \rho_j^T C_{\text{sym},n}}
        =2^{2n+1}-2^{n+1}\tr_{\C^{2^n}}[\rho_j \, \rho_k].   \nonumber
\end{align}
Therefore, $D_{\text{sym},n}(\rho_j,\rho_k)=2^{n+1/2}$ is equivalent to $\tr_{\C^{2^n}}[\rho_j \, \rho_k]=0$ which is the case as $\tr_{\C^{2^n}}[\rho_j \, \rho_k]=\abs{\inner{\psi_j}{\psi_k}}^2$ and $\psi_k$ is orthogonal to $\psi_j$ whenever $k \neq j.$
\par
To see the converse direction (ii) $\Rightarrow$ (i) recall that by Proposition \ref{prop:diameter-of-S-n}, two states can realize the diameter only if they are Hilbert-Schmidt orthogonal, that is, $D_{\text{sym},n}(\rho_j,\rho_k)=2^{n+1/2}$ implies that $\tr_{\C^{2^n}}[\rho_j \, \rho_k]=0.$ The Hilbert-Schmidt inner product of two positive operators is zero if and only if their ranges are orthogonal, and hence $\mathrm{ran}(\varrho_j)\perp\mathrm{ran}(\varrho_k)$ for every $j\neq k.$ Every density operator has positive rank, so the pairwise orthogonality of the $2^n$ density operators $\rho_0, \rho_1, \dots, \rho_{2^n-1}$ acting on $\C^{2^n}$ implies that $\dim(\mathrm{ran}(\varrho_j))=1$ for all $j \in {0,1, \dots, 2^n-1},$ and hence $\varrho_j\in \mathcal{P}_1(\mathbb{C}^{2^n}).$ In particular, $\rho=\rho_0$ is a pure state.
\end{proof}
The above characterization of pure states refers only to the symmetric quantum Wasserstein distance, and hence pure states are preserved by isometries, as stated in the following corollary. 

\begin{corollary} \label{cor:pure-states-are-preserved}
If $\Phi: \cS\ler{\C^{2^n}} \to \cS\ler{\C^{2^n}}$ is an isometry with respect to the quantum Wasserstein distance $D_{sym, n},$ then $\Phi\ler{\cP_1\ler{\C^{2^n}}} \subseteq \cP_1\ler{\C^{2^n}},$ that is, $\Phi(\rho)$ is a pure state whenever $\rho$ is so.    
\end{corollary} 

\begin{proof}
    By Proposition \ref{prop:metric-char-of-pure-states}, if $\rho \in \cP_1\ler{\C^{2^n}},$ then there exist $\rho_1, \dots, \rho_{2^n-1}$ such that $D_{\text{sym},n}(\rho_j,\rho_k)=2^{n+1/2}$ for all $j,k \in \lers{0,1,\dots, 2^n-1}, \, j \neq k.$ This implies that also $D_{\text{sym},n}\ler{\Phi(\rho_j),\Phi(\rho_k)}=2^{n+1/2}$ for all distinct $j$ and $k,$ as $\Phi$ is an isometry. Consequently, by Proposition \ref{prop:metric-char-of-pure-states} again, $\Phi(\rho) \in \cP_1\ler{\C^{2^n}}.$
\end{proof}

This means that the restriction $\Phi_{|\cP_1\ler{\C^{2^n}}}$ of a Wasserstein isometry $\Phi$ to the set of pure states is a $\cP_1\ler{\C^{2^n}} \to \cP_1\ler{\C^{2^n}}$ map that preserves the distance $D_{\text{sym},n}.$ However, on pure states, the Wasserstein distance $D_{\text{sym},n}$ takes the simple form 
\begin{align}
    D_{\text{sym},n}^2(\rho, \omega)=2^{2n+1}-2^{n+1}\tr_{\C^{2^n}}[\rho \, \omega] \qquad \ler{\rho, \omega \in \cP_1\ler{\C^{2^n}}}
\end{align}
as pure states admit only the tensor product coupling.
Therefore, the condition
\begin{align} 
     D_{\text{sym},n}^2\ler{\Phi(\rho), \Phi(\omega)}= D_{\text{sym},n}^2(\rho, \omega) \qquad \ler{\rho, \omega \in \cP_1\ler{\C^{2^n}}} \nonumber
\end{align}
is equivalent to
\begin{align} 
     \tr_{\C^{2^n}}\lesq{\Phi(\rho) \Phi(\omega)}= \tr_{\C^{2^n}}[\rho \, \omega] \qquad \ler{\rho, \omega \in \cP_1\ler{\C^{2^n}}}, \nonumber
\end{align}
which means that $\Phi_{|\cP_1\ler{\C^{2^n}}}: \cP_1\ler{\C^{2^n}} \to \cP_1\ler{\C^{2^n}}$ is a map that preserves the transition probability of states. Here we recall Wigner's famous theorem \cite{wigner1931gruppentheorie} (see also \cite{Lomont-Mendelson-Wigner, Bargmann-Wigner, Molnar-Wigner, Geher-Wigner}) describing the structure of transition probability preserving maps on pure states. As our work concerns finite-dimensional Hilbert spaces, but we work with a map which is not assumed to be injective or surjective, we state the finite-dimensional and non-bijective version of Wigner's theorem \cite{Geher-Wigner} for simplicity.

\begin{theoremext}[Wigner's theorem, finite-dimensional, non-bijective version]
If $\cH$ is a finite-dimensional Hilbert space, and $T: \cP_1(\cH) \to \cP_1(\cH)$ is a map that preserves the transition probability between pure states, that is, 
\begin{align}
    \tr_{\cH}[T(\rho) T(\omega)]=\tr_{\cH}[\rho \, \omega] \qquad \ler{\rho, \omega \in \cP_1\ler{\cH}}, \nonumber
\end{align}
then $T$ is a unitary or anti-unitary conjugation, that is, there exists a unitary or anti-unitary operator $U:\cH \to \cH$ such that
\begin{align}
    T(\rho)=U \rho U^* \qquad \ler{\rho \in\cP_1(\cH)}.    
\end{align}
\end{theoremext}

So by Wigner's theorem, for every transformation $\Phi: \cS\ler{\C^{2^n}} \to \cS\ler{\C^{2^n}}$ that preserves $D_{\text{sym},n}$ there exists a unitary or anti-unitary $U$ such that $\Phi_{|\cP_1\ler{\C^{2^n}}}$ is of the form $\Phi(\rho)=U \rho U^*.$
Now let $\Phi$ be an isometry of $\cS\ler{\C^{2^n}}$ with respect to $D_{\text{sym},n},$ and let $U$ be a unitary or anti-unitary operator such that $\Phi(\rho)=U \rho U^*$ for all $\rho \in \cP_1\ler{\C^{2^n}}.$ As we have seen in Proposition \ref{prop:Wigner-symmetries-are-isometries}, the map $\rho \mapsto U^* \rho U$ is a $D_{\text{sym},n}$-isometry of $\cS\ler{\C^{2^n}},$ and hence the composition $\widetilde\Phi$ defined by 
\begin{align} \label{eq:Phi-tilde-def}
    \widetilde\Phi(\rho):=U^* \Phi(\rho) U \qquad \ler{\rho \in \cS\ler{\C^{2^n}}}
\end{align}
is also a $D_{\text{sym},n}$-isometry with the notable property that $\widetilde\Phi(\rho)=\rho$ for all $\rho \in \cP_1\ler{\C^{2^n}}.$ We proceed by showing that $\widetilde\Phi$ is necessarily the identity on the whole state space $\cS\ler{\C^{2^n}}.$ Let $\rho \in \cP_1\ler{\C^{2^n}}$ and $\omega \in \cS\ler{\C^{2^n}}$ be arbitrary. Then $\cC(\rho, \omega)=\{\omega \otimes \rho^T\},$ and hence 
\begin{align} \label{eq:on-one-hand}
    D_{\text{sym},n}^2(\rho, \omega)=2^{2n+1}-2^{n+1}\tr_{\C^{2^n}}[\rho \omega]. 
\end{align}
Moreover, $\cC(\rho, \widetilde\Phi(\omega))=\lers{\widetilde\Phi(\omega) \otimes \rho^T},$ and consequently, using the isometric property of $\widetilde\Phi$ on $\cS\ler{\C^{2^n}}$ we get 
\begin{align} \label{eq:on-the-other-hand}
    D_{\text{sym},n}^2(\rho, \omega)
    =D_{\text{sym},n}^2(\widetilde\Phi(\rho), \widetilde\Phi(\omega))
    =D_{\text{sym},n}^2(\rho, \widetilde\Phi(\omega))=2^{2n+1}-2^{n+1}\tr_{\C^{2^n}}[\rho \widetilde\Phi(\omega)]. 
\end{align}
So \eqref{eq:on-one-hand} and \eqref{eq:on-the-other-hand} imply that $\tr_{\C^{2^n}}[\rho \widetilde\Phi(\omega)]=\tr_{\C^{2^n}}[\rho \omega]$ for all $\rho \in \cP_1\ler{\C^{2^n}}$ and $\omega \in \cS\ler{\C^{2^n}}.$ In other words, $\bra{\psi}\widetilde\Phi(\omega) \ket{\psi}=\bra{\psi}\omega \ket{\psi}$ for all unit vectors $\psi \in \C^{2^n}$ and all states $\omega \in \cS\ler{\C^{2^n}}.$
The quadratic form of a self-adjoint operator uniquely determines the operator itself, and hence we just deduced that $\widetilde\Phi(\omega)=\omega$ for all $\omega \in \cS\ler{\C^{2^n}}.$ That is, $\widetilde\Phi$ is indeed the identity of $\cS\ler{\C^{2^n}}.$ We defined $\widetilde\Phi$ as the composition of $\Phi$ and the unitary or anti-unitary conjugation $\rho \mapsto U^* \rho U,$ see \eqref{eq:Phi-tilde-def}, and hence $\widetilde\Phi$ being the identity of $\cS\ler{\C^{2^n}}$ is equivalent to 
\begin{align}
    \Phi(\rho)=U \rho U^* \qquad \ler{\rho \in \cS\ler{\C^{2^n}}}.
\end{align}
This means that every quantum Wasserstein isometry of $\cS\ler{\C^{2^n}}$ with respect to the distance $D_{\text{sym},n}$ is a unitary or anti-unitary conjugation, as desired, and hence the proof of Theorem \ref{thm:n-qubit-symmetric} is complete. 

\medskip

\paragraph*{{\bf Acknowledgment.}} We thank the anonymous referee for his/her valuable comments and insightful suggestions.

%%% Bibliography
\begin{small}
\bibliographystyle{plainurl}  % ama, nar, alpha, plain, chicago, abbrv, siam
\bibliography{references.bib}
\end{small}

\end{document}